\PassOptionsToPackage{dvipsnames}{xcolor}
\documentclass[letterpaper, 10 pt, conference]{ieeeconf}  

\IEEEoverridecommandlockouts                              

\usepackage{mathptmx} 
\usepackage{times} 
\usepackage{amsmath} 
\usepackage{amssymb}  
\usepackage{bm}
\usepackage{cite}
\usepackage{comment}
\usepackage{paralist}
\usepackage{pdfpages}
\usepackage{hyperref}
\usepackage{subcaption}

\usepackage{xspace}
\usepackage{xcolor}

\usepackage{amsthm}
\newtheorem{theorem}{Theorem}
\newtheorem{proposition}{Proposition}

\theoremstyle{remark}
\newtheorem{remark}{Remark}
\definecolor{pyplotblue}{RGB}{3,67,223}
\definecolor{pyplotred}{RGB}{229,0,0}

\newcommand{\AirSim}{AirSim\xspace}
\newcommand{\LandscapeMountains}{\textsf{LandscapeMountains}\xspace}
\newcommand{\AirSimNH}{\textsf{AirSimNH}\xspace}
\newcommand{\ZhangJiajie}{\textsf{ZhangJiajie}\xspace}

\newcommand{\PercCont}{perception contract\xspace}

\newcommand{\Real}{\ensuremath{\mathbb{R}}\xspace}

\newcommand{\Fcam}{\ensuremath{\mathit{s}}\xspace}
\newcommand{\Fnn}{\ensuremath{\mathit{h}}\xspace}

\newcommand{\round}[1]{\ensuremath{\left\lfloor #1 \right\rceil}\xspace}

\newcommand{\pos}[1]{\ensuremath{{q}_{#1}}\xspace}

\newcommand{\img}[1]{\ensuremath{{m}_{#1}}\xspace}

\newcommand{\tgt}[1]{\ensuremath{{q}_{#1}^*}\xspace}

\newcommand{\stateq}{\ensuremath{\mathrm{q}}\xspace}
\newcommand{\targetq}{\ensuremath{\mathrm{q}^*}\xspace}

\newcommand{\statey}{\ensuremath{\mathrm{y}}\xspace}
\newcommand{\targety}{\ensuremath{\mathrm{y}^*}\xspace}

\newcommand{\statez}{\ensuremath{\mathrm{z}}\xspace}

\newcommand{\Eqstar}{\ensuremath{\mathrm{E}_{\targetq}}\xspace}
\newcommand{\Safe}{\ensuremath{\mathrm{S}}\xspace}

\newcommand{\Nodes}{\ensuremath{\{1\dotsc N\}}\xspace}

\newcommand{\CamM}{\ensuremath{K}\xspace}
\newcommand{\RotM}{\ensuremath{R}\xspace}
\newcommand{\tvec}{\ensuremath{{t}}\xspace}

\newcommand{\chiaocolor}{\color{black}}
\newenvironment{chiaoenv}{\chiaocolor}{}
\newcommand{\chiao}[1]{{\chiaocolor #1}}
\newcommand{\sayan}[1]{\textcolor{black}{#1}}

\title{
Assuring Safety of Vision-Based Swarm Formation Control
}

\author{%
Chiao Hsieh\textsuperscript{1},
Yubin Koh\textsuperscript{1},
Yangge Li\textsuperscript{1}
and Sayan Mitra\textsuperscript{1}
\thanks{$^{1}$The authors are with Coordinated Science Laboratory, University of Illinois Urbana-Champaign,
        Champaign, IL, USA
        \texttt{\{chsieh16, yubink2, li213, mitras\}@illinois.edu}}%
}

\begin{document}

\maketitle
\thispagestyle{empty}
\pagestyle{empty}

\begin{abstract}

Vision-based formation control systems are attractive because they can use inexpensive sensors and can work in GPS-denied environments. The safety assurance for such systems is challenging: the vision component's accuracy depends on the environment in complicated ways, these errors propagate through the system and lead to incorrect control action, and there exists no formal specification for end-to-end reasoning.
We address this problem and propose a technique for safety assurance of vision-based formation control:
First, we propose a  scheme for constructing quantizers that are consistent with vision-based perception.  Next, we show how the convergence analysis of a standard quantized consensus algorithm can be adapted for the  constructed quantizers. 
We use the recently defined notion of \emph{perception contracts} to create error bounds on the actual vision-based perception pipeline using sampled data from different ground truth states, environments, and weather conditions. 
Specifically, we use a quantizer in logarithmic polar coordinates, and we show that this quantizer is sutiable for the constructed perception contracts for the vision-based position estimation, where the error worsens with respect to the absolute distance between agents. We build our formation control algorithm with this nonuniform quantizer, and we prove its convergence employing an existing result for quantized consensus.

\end{abstract}

\section{Introduction}

Distributed consensus, flocking, and formation control have been studied extensively, including in scenarios where the  participating agents only have partial state information (see, for example~\cite{Blondel,Saber.Murray2003Flockingwithobstacle,mesbahi2010graph,bullo2009distributed}). 
With the advent of deep learning and powerful computer vision algorithms, it is now feasible for agents to use vision-based state estimation for formation control~(See Figure~\ref{fig:airsim-drones-cameras}). Such systems can be attractive because they do not require expensive sensors and localization systems, and also can be used in GPS-denied environments~\cite{montijano2016vision,fathian2018distributed,fallah2022visual}. However, deep learning and vision algorithms are well-known to  be fragile, which can break the correctness and safety of the end-to-end formation control system. Further, it is  difficult to specify the correctness of a vision-based state estimator, which gets in the way of modular design and testing of the overall  formation control system~\cite{emsoft2022_industry}. In this paper, we address these challenges and present the first end-to-end   formal analysis of a vision-based formation control system.


\begin{figure}[ht]
    \centering
    \includegraphics[height=4cm]{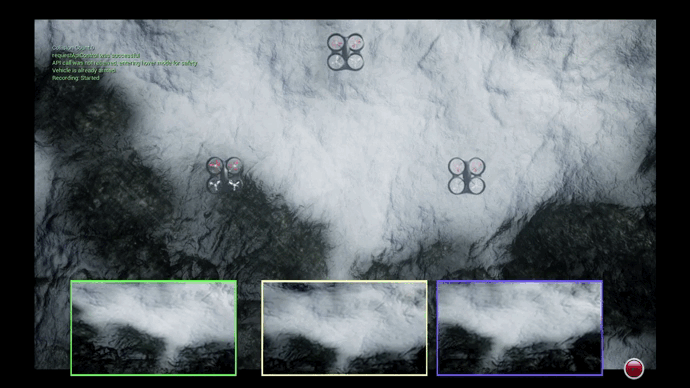}
    \caption{Vision-based drone formation using downward facing camera images in AirSim.}
    \label{fig:airsim-drones-cameras}
\end{figure}

We present analyses for both convergence and safety assurance of a vision-based swarm formation control system.
The computer vision pipeline (See Figure~\ref{fig:architecture}) uses feature detection, feature matching,  and geometry to estimate the relative position of the participating drones.
The estimated relative poses are then used by a consensus-based formation control algorithm.
There are two key challenges in analyzing the system:
\begin{inparaenum}[(1)]
\item The perception errors impact the behavior of neighboring agents and propagates through the entire swarm.
\item The magnitude of the perception error is highly nonuniform, and depends on the ground truth values of the relative position between neighboring agents.
\end{inparaenum}
In general, perception errors can get worse as the system approaches the equilibrium (desired formation), and thus, make stabilization difficult.
Environmental variations (e.g., lighting, fog) are other factors that can make the vision-based system unstable.

\sayan{
In addressing the problem, our idea is to view the vision-based formation control system as a \emph{quantized consensus protocol}~\cite{kashyap2007quantized}.
We start with the \emph{assumption} that the impact of the state estimation errors arising from the vision pipeline can be encapsulated as quantization errors in a non-uniform quantization scheme.
\chiao{That is, the quantization step size can vary non-uniformly with respect to the state so that the quantization errors can overapproximate state dependent perception errors.}
To discharge this assumption, our analysis has to meet two requirements.
First, we have to propose a \emph{specific} quantization scheme under which the formation control system is indeed guaranteed convergence. For this, we develop a quantized formation controller~(Equation~\eqref{eq:controller_yi}) and a \emph{logarithmic polar quantizer}~(Equation~\eqref{def:quantized-radius} and \eqref{def:quantized-angle}),
and we show in Theorem~\ref{thm:y-eq} that indeed the resulting quantized formation control protocol converges, using sufficient conditions from~\cite{kashyap2007quantized}.
Secondly, we have to show that a quantizer instantiated from the quantization scheme matches the error characteristics of the vision pipeline. For this part, we utilize the recently developed idea of \emph{\PercCont{}s}~\cite{hsieh2022aap,astorga2023perception}. A perception contract~(PC) for a vision-based state estimator bounds the estimation error as a function of the ground truth state. Earlier in~\cite{astorga2023perception}, PCs have been used to establish the safety of vision-based lane keeping systems. For formation control, however, the PCs are dramatically different because the error has a highly non-uniform dependency on the state; as the drones get closer, the error drops.
Through data-driven construction of the logarithmic PC,
we show that the vision pipeline indeed matches the PC with high probability in Section~\ref{sec:perc-error},
and we further adapt to environmental variations by inferring quantization step sizes for different environments.}

In summary, our contributions are as follows:
\begin{inparaenum}[(1)]
\item An approach to construct a quantizer as the \PercCont of the vision component.
\item \sayan{Empirical} analysis of the impact of environmental variations on the \PercCont with the photorealistic \AirSim simulator~\cite{shah2018airsim}.
\item \sayan{Theoretical analysis of the overall formation control system using the constructed quantizer, which gives the bounds on the convergence time.
}
\end{inparaenum}
Our code of the vision pipeline, simulation script, and analysis tool are publicly available\footnote{%
Repository: \url{https://gitlab.engr.illinois.edu/aap/airsim-vision-formation}}.

\begin{figure}[t]
    \centering
    \includegraphics[width=\columnwidth]{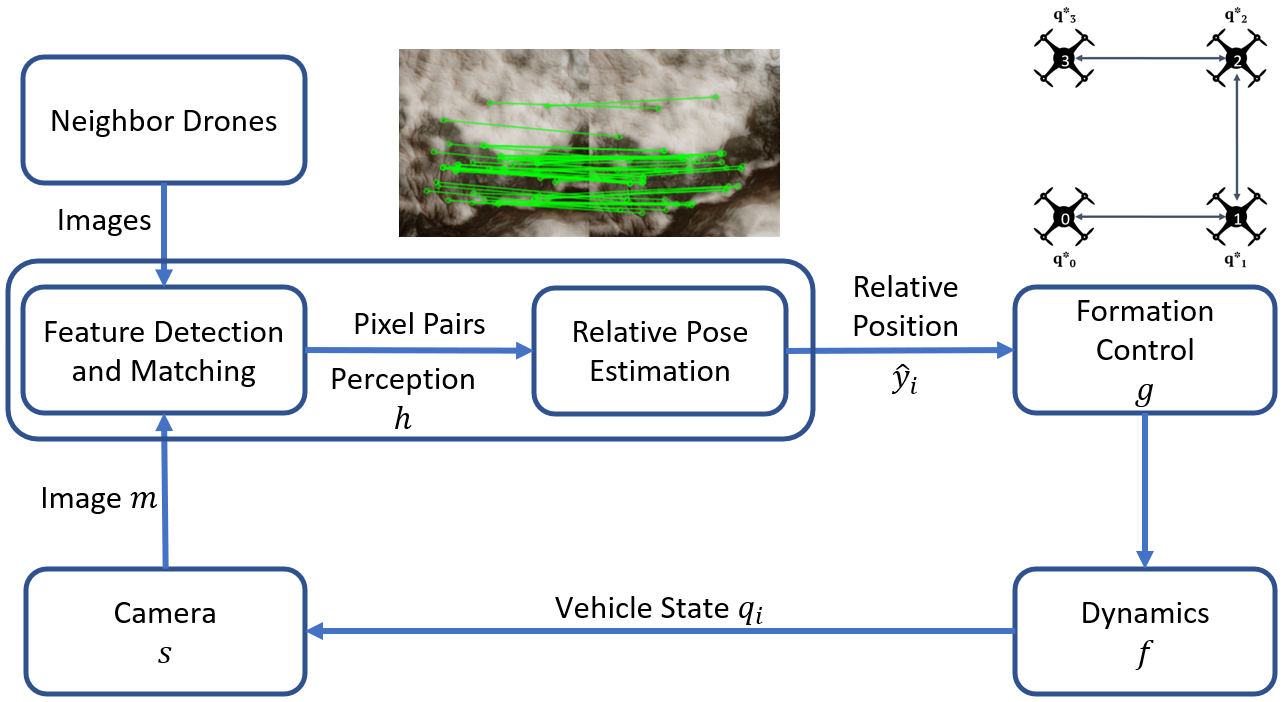}
    \caption{Architecture of an agent in the vision-based formation control systems.}
    \label{fig:architecture}
\end{figure}

\paragraph*{Related Works}
Two parallel threads of research have recently addressed formal end-to-end analysis of vision-based autonomous systems. The works in~\cite{Pasareanu-CAV23,Pasareanu22-controllerRV,Pasareanu2023assumption} approach the problem using discrete state space and stochastic models. Our previous work~\cite{hsieh2022aap,astorga2023perception} develops the idea of {\em perception contracts} using the language of continuous state space models. Thus far all the applications studies in these threads are related to lane following by a single agent, which is quite different from distributed formation control.

VerifAI~\cite{verifAI} uses techniques like fuzz testing and simulation to falsify the system specifications.
Katz et al.~\cite{katz2021verification} trains generative adversarial networks (GANs) to produce a network to simplify the image-based NN.
NNLander-VeriF~\cite{cruz2022nnlander} verifies NN perception along with NN controllers for an autonomous landing system.
In contrast, 
our current work aims to provide safety analyses for a formation control system with vision-based perception,
and we apply the analysis on convergence to quantized consensus~\cite{kashyap2007quantized} for safe separation and formation.

\paragraph*{Paper Organization}
In Section~\ref{sec:prelim}, we introduce the formation control system with the vision-based perception and review the quantized formation controller.
In Section~\ref{sec:ult-bound}, we show the convergence under perception error using our main theory of quantized consensus.
In Section~\ref{sec:perc-error}, we describe the quantization for perception error bounds via sampling from vision-based pose estimation with AirSim simulation.
We then conclude in Section~\ref{sec:conclusion}.

\section{Vision-Based Formation Control}\label{sec:prelim}

We will study a distributed formation control system
with $N+1$ identical aerial vehicles or \emph{agents} with a leader agent 0 as shown in Figure~\ref{fig:airsim-drones-cameras}.
The target formation is specified in terms of relative positions between agents.
Each agent $i$ has a downward facing camera, and it uses images from its own camera and its predecessor $i-1$'s camera to periodically estimate the relative position of $i$ with respect to $i-1$.
Based on the estimated relative positions to its neighbor, agent $i$ then updates it own position by setting a velocity, to try and achieve the target formation.

Before describing the vision and control modules in more detail, we introduce some notations.
First, we describe the neighborhood relation between agents by an undirected connected graph $G = (V, E)$, where $V = \{0,1,\dotsc,N\}$.
Second, we only consider planar formations for simplicity though the agents are in 3-dimensional space.
Thus, the position of agent $i$ in the world frame is represented by a vector $\pos{i} \in \Real^2$.
The state of the overall system is a sequence $\stateq = (\pos{0}, \pos{1}, \dotsc, \pos{N})$.
The distributed formation control system evolves with a goal of reaching a target formation in a set $\Eqstar$ that is specified by a vector $\targetq = (\pos{0}^*, \pos{1}^*, \dotsc, \pos{N}^*)$ as
\[
\Eqstar = \{\stateq \mid \forall i \in \{1, \dotsc, N \}, \pos{i} - \pos{i-1} = \tgt{i} - \tgt{i-1}\}.
\]
That is, \Eqstar is the set of all states that form \targetq up to  translations.
We also specify a {\em safe set\/}  that the where the distance between no two agents  is too close or too far:
\[
\Safe_\stateq = \{\stateq \mid \forall i \in \{1, \dotsc, N \}, d_{\min} < \lVert\pos{i} - \pos{i-1}\rVert < d_{\max}\}
\]
where $0 < d_{\min} < d_{\max}$ defines the range of safe distances.

\subsection{Vision-Based Relative Pose Estimation}
\label{subsec:prelim-vision}

\begin{figure}[t!]
    \includegraphics[width=\columnwidth, height=25mm, keepaspectratio]{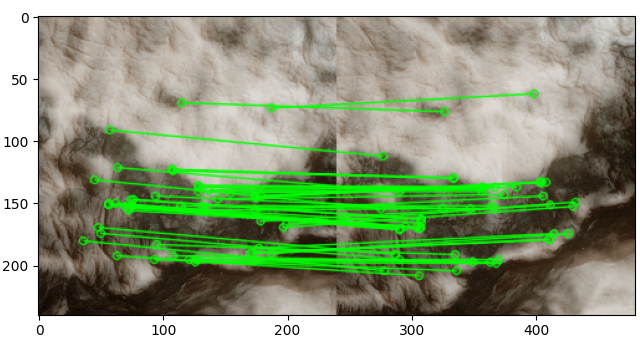}
    \centering
    \caption{Feature matching on an image pair from \AirSim.}
    \label{fig:feature-matching}
\end{figure}

We now discuss the components of an agent $i$ (Figure~\ref{fig:architecture}).
Agent $i$'s downward-facing camera $\Fcam$ periodically generates an image of the ground $\img{i}$, which depends on its state $\pos{i}$ and other environmental factors like background scenery, lighting, fog, etc.
The neighboring agent $i-1$ generates another image $\img{i-1}$ of the ground and shares this with agent $i$ over the communication channel. We assume the whole system runs synchronously in lock-step, i.e., both neighboring drones will capture the image at the same time and there's no communication delay between drones while sharing images.
The vision-based pose estimation algorithm $\Fnn$ takes a pair of images, $\img{i-1}$ and $\img{i}$, as an input and produces the \emph{estimated relative position} $\hat{y}_{i}$ to estimate the relative position of agent $i$ with respect to agent $i-1$, i.e., $\pos{i} - \pos{i-1}$.
The estimation algorithm in general follows these steps:
(1) First, $\Fnn$ detects features from each image. Any of the various feature detection algorithms like SIFT~\cite{lowe04sift}, SURF~\cite{bay2006surf}, and ORB~\cite{2011-rublee} can be used for this step.
(2) Then, $\Fnn$ collects the detected features from the pair of images, and a feature matching algorithm (such as FLANN~\cite{muja09flann}) is used to match  pairs of features in each image as shown in Figure \ref{fig:feature-matching}.
(3) For each feature point, the relationship between the pixel coordinate and the world coordinate of the feature point~\cite{ma2003invitation} is described by 
\(
    s[u\ v\ 1]^T = \CamM[\RotM\mid \tvec][X\ Y\ Z\ 1]^T
\)
where $[u\ v\ 1]^T$ is the pixel coordinate, $[X\ Y\ Z\ 1]^T$ is the world coordinate of each detected feature, \CamM is the camera intrinsic matrix and $[\RotM\mid \tvec]$ is the extrinsic camera parameters.
With a set of at least 8 matched features, we can come up with 8 pairs of equations between the poses of two cameras, and by solving these equations, we can calculate the relative rotation and the normalized translation vector using the inverse geometry of image formation.
Examples of this step appear in~\cite{Malis2007DeeperUO,nister2004relativepose, li2006fivepoint}.
Further, the altitude information and drone orientation can be used to estimate the true distance to ground and recover the length of the translation vector.

The accuracy of the perception pipeline can be influenced by many factors.
The change of environments such as background, lighting, and weather influence the quality of the image and the image features, which in turn influence the accuracy of relative pose estimation.

\subsection{Formation Control in Relative Positions}
\label{ssec:formation}

To simplify the notations, let $Y$ denote the vector space of the \emph{relative positions} between pairs of drones. Let $y_i \in Y$ be defined as $y_i = \pos{i} - \pos{i-1}$, for $i = 1, \dotsc, N$.
A state is a sequence of relative positions $\statey = (y_1, y_2,\dotsc,y_N)$.
Let $y_i^* \in Y$ be the desired relative vector between drone $i-1$ and $i$,
i.e., $y_i^* = \pos{i}^* - \pos{i-1}^*$.
The target equilibrium state is defined by:
\[
\targety = (y_1^*,\dotsc,y_N^*),
\]
and the safe set is
\begin{equation*}
\Safe_\statey := \{\statey \mid \forall i \in \{1, \dotsc, N\}, d_{\min} < \lVert y_i \rVert < d_{\max}\}.    
\end{equation*}
The following proposition relating the \statey-system with the original \stateq-system follows immediately. 
\begin{proposition}
The formation control system (\stateq-system) reaches a desired state in \Eqstar if and only if the system in relative positions (\statey-system) reaches the target \targety.
In addition, the formation control system (\stateq-system) stays within the safe set $\Safe_\stateq$ if and only if the system in relative positions (\statey-system) stays within the safe set $\Safe_\statey$.
\end{proposition}

\subsection{Quantized Formation and Perception Contract}
\label{ssec:quant}
\begin{chiaoenv}
Given the system in relative positions with the state at time $t$ represented by $\statey[t]=(y_1[t],\dotsc,y_N[t])$,
we aim to handle perception errors by designing a \emph{quantizer} $Q$ and a quantized formation controller built on top of $Q$.
The insight is as follows: If the ground-truth $y_i[t]$ and the perceived relative position $\hat{y}_i[t]$ always lead to the same quantized value after quantization, i.e., $Q(y_i[t]) = Q(\hat{y}_i[t])$,
then any quantized formation controller using $Q(\hat{y}_i[t])$ instead of $Q(y_i[t])$ will still stabilize the system regardless of perception errors.

More precisely,
we follow the definitions in~\cite{liberzon2003hybrid} and choose a subset of positions $W = \{w_1, w_2, \dotsc\} \subset Y$ selected for quantization.
A \emph{quantizer} is a function $Q: Y \to W$ which partitions $Y$ into \emph{quantization regions} of the form $\{y \in W \mid Q(y)=w_k\}$ for each $w_k \in W$.
Given a target state $\targety=(y_1^*,\dotsc,y_N^*)$ with all $y_i^* \in W$,
a \emph{quantized formation controller} ensures that the system is evolving between states with all $y_i[t] \in W$ for all time $t$,
and stabilizes the system to the target \targety.
The \emph{perception contract}~(\emph{PC}) then requires that the perceived value is always in the quantization region defined by the ground-truth quantized value, that is, $\hat{y}_i[t] \in \{y \mid Q(y) = y_i[t]\}$ or equivalently $Q(\hat{y}_i[t]) = Q(y_i[t])$.
\end{chiaoenv}

\newcommand{\mean}{\ensuremath{\mathit{mean_\omega}}\xspace}
\newcommand{\wrap}{\ensuremath{\mathit{wrap}}\xspace}
\newcommand{\sub}{\ensuremath{\mathit{diff}}\xspace}

In this paper, we study a particular template of quantized formation controllers constructed using a quantizer, a difference function, and a weighted average function.
We assume a generalized difference function \sub to calculate the difference between two relative positions and a function \mean to compute a weighted midpoint parametrized by a real value $\frac{1}{2} < \omega < \frac{3}{4}$.
The system evolves according to a (nondeterministic) discrete dynamical system
with a pair $(i,j)$ selected randomly at each time step $t$,
and the quantized formation controller updates the pair of states as follows:
\begin{equation}\label{eq:controller_yi}
\begin{small}
\begin{aligned}
    y_i[t+1] &= Q(\mean(Q(y_i[t]), \sub(Q(y_j[t]), \sub(y_j^*, y_i^*)))) \\
    y_j[t+1] &= Q(\mean(Q(y_j[t]), \sub(Q(y_i[t]), \sub(y_i^*, y_j^*)))) \\
\end{aligned}
\end{small}
\end{equation}
Note that, by design, the updated states are always quantized values given any \sub and \mean,
and they are unaffected by perception errors whenever the perception contract holds.

\begin{chiaoenv}
This template allows us to design a quantized formation controller with vision-based perception in two separate steps:
\begin{inparaenum}[(1)]
\item find sufficient conditions for the quantizer that ensures the convergence and safety of the quantized formation controller and
\item construct the quantizer from observed perception errors to serve as the perception contract.
\end{inparaenum}
In Section~\ref{sec:ult-bound},
we derive the sufficient conditions for the safety and convergence of the quantized formation.
We prove that, given a \emph{quantizer} over \emph{logarithmic polar}~(log-polar) coordinates,
the system in Equation~\eqref{eq:controller_yi} safely converges to the target state,
and the expected value for the time of convergence is bounded.
In Section~\ref{sec:perc-error}, we study the empirically observed perception error,
and we explain how to design a quantizer to approximate the perception error.
\end{chiaoenv}
\section{Convergence of Quantized Formation}\label{sec:ult-bound}

\newcommand{\etap}{\ensuremath{\Tilde{\eta}}\xspace}
\newcommand{\Lambdap}{\ensuremath{\Tilde{\Lambda}}\xspace}

In this section, we prove that the true relative positions between agents, with a quantizer on log-polar coordinates modeling the perception error,
converges to the target formation and stays within the safe distance bounds.
This is done by proving that the system in Equation~\eqref{eq:controller_yi} is simulated by a \emph{quantized averaging} algorithm in~\cite{kashyap2007quantized},
which is proven to always converge to quantized consensus and stay within a bounded interval.
In addition, the expected convergence time of the algorithm is bounded.
We first provide the quantizer in Section~\ref{subsec:quantizer}.
We then prove the simulation relation between formation control systems and quantized averaging algorithms in Section~\ref{subsec:translation},
and the bound on the convergence time in Section~\ref{subsec:timebound}.

\subsection{Quantization on Logarithmic Polar Coordinates}\label{subsec:quantizer}
We first define the set of selected positions $W$ in polar coordinates.
Without loss of generality,
we choose a quantization \emph{step radius} $a > 1$ and define the set of quantized radii $R = \{a^0, a^{\pm 1}, a^{\pm 2}, \dotsc\}$ and a \emph{step angle} $\theta_b = \frac{2\pi}{M}$ to define the set of quantized angles $\Theta = \{0, \theta_b, \dotsc, (M-1)\cdot\theta_b\}$ with an integer $M \geq 2$;
then we define the set of selected points $W$ by $W=R\times\Theta$.
Equivalently, given two positions in polar coordinates $y_i= [r_i\,\theta_i]^T$ and $y_j=[r_j \, \theta_j]^T$ with angles normalized to $[0, 2\pi)$,
we define the quantizer and other functions for the radial coordinate as below:
\begin{equation}\label{def:quantized-radius}
\begin{gathered}
I(r_i) = \round{\log_{a}r_i}; \quad
Q(r_i) = a^{I(r_i)}; \\
\sub(r_i, r_j) = \frac{r_i}{r_j}; \quad
\mean(r_i, r_j) =  r_i^{(1-\omega)} \cdot r_j^\omega
\end{gathered}
\end{equation}
where \round{\cdot} rounds the real number to the nearest integer.
The functions for the angular coordinate are as below:
\begin{equation}\label{def:quantized-angle}
\begin{gathered}
I(\theta_i) \equiv_M \round{\frac{\theta_i}{\theta_b}}; \quad
Q(\theta_i) = I(\theta_i)\cdot\theta_b; \\
\sub(\theta_i, \theta_j) = \theta_i \oplus -\theta_j; \\
\mean(\theta_i, \theta_j) = \theta_i \oplus (\omega \cdot (\theta_j \oplus -\theta_i))
\end{gathered}
\end{equation}
where $\equiv_M$ means congruent modulo $M$, $\oplus$ is the addition in the commutative group $SO(2)$, and $-\theta_i$ represents the additive inverse of $\theta_i$.
The \mean function calculates the weighted geometric mean of rotations~\cite{moakher2002means}.

\subsection{Simulation by Quantized Averaging Algorithms}\label{subsec:translation}
Following~\cite{kashyap2007quantized}, let $\statez[t] = (z_1[t], \dotsc, z_N[t])$ denote the state in the integer domain at each time step.
An instance of quantized averaging algorithms evolves as Equation~\eqref{eq:controller-z} below:
\begin{equation}\label{eq:controller-z}
\begin{aligned}
    z_i[t+1] &= z_i[t] + \round{\omega\cdot(-z_i[t] + z_j[t])} \\ 
    z_j[t+1] &= z_j[t] + \round{\omega\cdot(z_i[t] - z_j[t])}
\end{aligned}
\end{equation}
The following proposition is directly from~\cite[Section 5]{kashyap2007quantized}.
\begin{proposition}\label{prop:quantized-average}
The system in Equation~\eqref{eq:controller-z} always converges to the set of equilibria:

\noindent\(
\xi = \{ (z_1,\dotsc, z_N) \mid z_i \in \{L, L+1\}, i\in\Nodes \sum\limits_{i=1}^N z_i = S\}
\)
where $S = \sum_{i=1}^N z_i[0]$ is the initial sum of the system and $L = \lfloor \frac{S}{N} \rfloor$ is the quantized average.   
\end{proposition}
We now prove the convergence of the formation control system using Proposition~\ref{prop:quantized-average}.

\vspace{0.5em}
\begin{theorem}\label{thm:y-eq}
For any initial state $\statey[0] \in \mathrm{Y}_0$ where
\begin{small}
\[
\mathrm{Y}_0 = 
\left\{ (\begin{bmatrix}
    r_1 \\
    \theta_1
\end{bmatrix},\dotsc,\begin{bmatrix}
    r_N \\
    \theta_N
\end{bmatrix}) \right|\left.\prod\limits_i Q(r_i) = \prod\limits_i r_i^* \land \bigoplus\limits_i Q(\theta_i)=\bigoplus\limits_i\theta_i^*
\right\},
\]
\end{small}%
the formation control system in Equation~\eqref{eq:controller_yi} converges to the target 
$\targety = (y_1^*, \dotsc, y_N^*)$ with $y_i^* \in W$ for all $i$.
\end{theorem}
\begin{proof}
The proof is to show that the formation control system in Equation~\eqref{eq:controller_yi} is simulated by the quantized averaging system in Equation~\eqref{eq:controller-z} for every time step,
and thus we guarantee the convergence using Proposition~\ref{prop:quantized-average}.
Recall in Section~\ref{subsec:quantizer} that the quantizer $Q$ has a corresponding indexing function $I$.
Let $y_i[t] = [r_i[t]\, \theta_i[t]]^T$ and $y_i^* = [r_i^*\, \theta_i^*]^T$,
we denote $z_i[t]_r$ for the radial coordinate and  $z_i[t]_\theta$ for the angular coordinate,
and we relate the two systems by $z_i[t]_r = I(r_i[t]) - I(r_i^*)$ and 
$z_i[t]_\theta = I(\theta_i[t]) - I(\theta_i^*)$ for all $i$.

We first derive for the radial coordinates.
Given the relation $z_i[t] = I(y_i[t]) - I(y_i^*)$ for all $i$,
we apply the indexing function $I$ on the new state $r_i[t+1]$ to calculate $z_i[t+1]_r$.
\begin{small}
\begin{align*}
I(r_i[t+1])
  &= I(Q(\mean(Q(r_i[t]), \sub(Q(r_j[t]), \sub(r_j^*, r_i^*))))) \\
  =& \log_a  a^\wedge\round{\log_a \left( Q(r_i[t])^{(1-\omega)} \cdot (\frac{Q(r_j[t])\cdot r_i^*}{r_j^*})^{\omega}\right) } \\
  =& \lfloor (1-\omega)\cdot(\log_a Q(r_i[t])) \\
   &\ +
    \omega\cdot(\log_a Q(r_j[t]) - \log_a r_j^* + \log_a r_i^*) \rceil \\
   &\because r_i^*, r_j^* \text{ are quantized values.} \\
  =& \lfloor
    (1-\omega)\cdot I(r_i[t])\ +
    \omega\cdot(I(r_j[t]) - I(r_j^*) + I(r_i^*)) \rceil \\
  =&\ I(r_i[t]) +
      \lfloor -\omega\cdot(I(r_i[t]) - I(r_i^*)) + \omega\cdot(I(r_j[t]) - I(r_j^*)) \rceil
\end{align*}
\end{small}%
We now calculate $z_i[t+1]_r$ as follows:
\begin{align*}
z_i[t+1]_r &= I(r_i[t+1]) - I(r_i^*) \\
  &= I(r_i[t]) - I(r_i^*) \\
  &\quad + \lfloor -\omega\cdot(I(r_i[t]) - I(r_i^*)) + \omega\cdot(I(r_j[t]) - I(r_j^*))\rceil \\
  &= z_i[t]_r + \round{\omega\cdot(-z_i[t]_r + z_j[t]_r)}
\end{align*}
This is exactly the same as the quantized averaging algorithm in Equation~\eqref{eq:controller-z}.
Therefore, we have shown that the indices of the quantized radial coordinates evolve according to the quantized averaging algorithm.

Similarly, we derive for the angular coordinate.
Recall the definitions in Equation~\eqref{def:quantized-angle}.
Given two quantized angles $\theta_i = Q(\theta_i)$ and $\theta_j = Q(\theta_j)$,
we can distribute the indexing function over the addition and inverse as follows:
\begin{gather*}
I(\theta_i \oplus \theta_j) \equiv_M \round{\frac{(I(\theta_i) + I(\theta_j))\cdot\theta_b}{\theta_b}}
                       \equiv_M I(\theta_i) + I(\theta_j) \\
I(-\theta_i) \equiv_M I(-Q(\theta_i)) \equiv_M \round{\frac{-I(\theta_i)\cdot \theta_b}{\theta_b}} \equiv_M -I(\theta_i)
\end{gather*}
Additionally, when $\theta_i = Q(\theta_i)$,
we can derive for $\frac{1}{2} < \omega < \frac{3}{4}$:
\[
I(\omega\cdot\theta_i) \equiv_M \round{\frac{\omega\cdot I(\theta_i)\cdot\theta_b}{\theta_b}} \equiv_M \round{\omega\cdot I(\theta_i)}
\]
We now calculate the index of the new state $\theta_i[t+1]$.
\begin{small}
\begin{align*}
  I(\theta_i[t+1])
  =\ & I(Q(\mean(Q(\theta_i[t]), \sub(Q(\theta_j[t]), \sub(\theta_j^*, \theta_i^*)))))\\
  \equiv_M\ & I(Q(\theta_i[t]) \oplus (\omega \cdot (Q(\theta_j[t]) \oplus -\theta_j^* \oplus \theta_i^* \oplus -Q(\theta_i[t])))) \\
  \equiv_M\ & I(\theta_i[t]) + I(\omega \cdot (-Q(\theta_i[t])\oplus\theta_i^* \oplus Q(\theta_j[t]) \oplus -\theta_j^*))) \\
  \equiv_M\ & I(\theta_i[t]) + \round{\omega\cdot(-I(\theta_i[t]) + I(\theta_i^*) + I(\theta_j[t]) - I(\theta_j^*))}
\end{align*}
\end{small}%
Then, we calculate $z_i[t+1]_\theta$:
\begin{align*}
z_i[t+1]_\theta=\ & I(\theta_i[t+1]) - I(\theta_i^*) \\
\equiv_M\ & I(\theta_i[t]) - I(\theta_i^*) \\
          & + \round{\omega\cdot(-I(\theta_i[t]) + I(\theta_i^*) + I(\theta_j[t]) - I(\theta_j^*))} \\
\equiv_M\ & z_i[t]_\theta + \round{\omega\cdot(-z_i[t]_\theta + z_j[t]_\theta)}
\end{align*}
Again, the index of the angular coordinate evolves exactly following the quantized averaging algorithm in Equation~\eqref{eq:controller-z}.
The derivation for $z_j[t+1]$ is the same and skipped.

Furthermore, for any initial state $\statey[0] \in \mathrm{Y}_0$,
we can derive the initial sum of the quantized averaging system as follows:
\begin{align*}
            & \prod_i Q(r_i[0]) = \prod_i r_i^* \ \land\   \bigoplus_i Q(\theta_i[0])=\bigoplus_i\theta_i^* \\
\Rightarrow & \prod_i \frac{Q(r_i[0])}{r_i^*} = 1 \ \land\  \bigoplus\limits_i (Q(\theta_i[0]) \oplus -\theta_i^*) = 0 \\
\Rightarrow & \sum_i \log_a Q(r_i[0]) - \log_a r_i^* = 0 \ \land\  \sum_i \frac{Q(\theta_i[0])}{\theta_b} - \frac{\theta_i^*}{\theta_b} = 0\\
\Rightarrow & \sum_i I(r_i[0]) - I(r_i^*) = 0  \ \land\  \sum_i I(\theta_i[0]) - I(\theta_i^*) = 0 \\
\Rightarrow & \sum_i z_i[0]_r = 0 \ \land\  \sum_i z_i[0]_\theta = 0
\end{align*}
As a result, the initial sum $S$ in the z-system is 0,
the quantized average $L$ is always 0,
and according to Proposition~\ref{prop:quantized-average} the z-system converges to the only equilibrium where all $z_i = L = 0$.
Therefore, $y_i$ converges to $y_i^*$ for all $i \in \{1, \dotsc, N\}$.
Thus, our formation control system converges to \targety.
\end{proof}

We further provide the initial condition such that the system will remain in the safe set.
\begin{theorem}\label{thm:y-safe}
Given a target state in the safe set $\targety \in \Safe_\statey$,
if the system starts in an initial state
$\statey[0] \in (\mathrm{Y}_0 \cap \Safe_0)$ where
\begin{small}
\[
\Safe_0 = 
\left\{ (\begin{bmatrix}
    r_1 \\
    \theta_1
\end{bmatrix},\dotsc,\begin{bmatrix}
    r_N \\
    \theta_N
\end{bmatrix}) \right|\left.
\bigwedge_{i=1}^N \frac{Q(d_{\min})}{\min\limits_j \{r_j^*\}}  < \frac{Q(r_i)}{r_i^*} < \frac{Q(d_{\max})}{\max\limits_j \{r_j^*\}}
\right\},
\]
\end{small}%
then the system will always stay in the safe set, i.e., $d_{\min} < r_i[t] < d_{\max}$ for all time $t$.
\end{theorem}
\begin{proof}
Here we show the proof steps for the minimum safe distance $d_{\min}$.
We first derive the lower bound for $z_i[0]_r$ from the initial condition $\statey[0] \in (\mathrm{Y}_0 \cap \Safe_0)$.
\begin{align*}
             & \frac{Q(d_{\min})}{\min\limits_j \{r_j^*\}}  < \frac{Q(r_i[0])}{r_i^*} \\
\Rightarrow\ & \round{\log_a d_{\min}} - \min\limits_j \{\log_a r_j^*\} < \round{\log_a r_i[0]} - \log_a r_i^*\\
\Rightarrow\ & I(d_{\min}) - \min\limits_j \{I(r_j^*)\} < I(r_i[0]) - I(r_i^*) = z_i[0]_r
\end{align*}
Then, from~\cite[Theorem 2]{kashyap2007quantized}, we know $z_i[t]_r$ is bounded by the minimum initial value, i.e., $\min_j \{z_j[0]_r\} \leq z_i[t]_r$ for all time $t$.
We use it to derive the lower bound on $r_i[t]$ for all time $t$ as follows.
\begin{align*}
            & I(d_{\min}) - \min\limits_j \{I(r_j^*)\} < \min_j \{z_j[0]_r\} \leq z_i[t]_r \\
\Rightarrow & I(d_{\min}) - \min\limits_j \{I(r_j^*)\} < I(r_i[t]) - I(r_i^*) \\
\Rightarrow & I(d_{\min}) + (I(r_i^*) - \min\limits_j \{I(r_j^*)\}) < I(r_i[t]) \\
\Rightarrow & I(d_{\min}) < I(r_i[t]) \qquad \because I(r_i^*) \geq \min\limits_j \{I(r_j^*)\} \\
\Rightarrow & Q(d_{\min}) < Q(r_i[t]) \Rightarrow d_{\min} < r_i[t]
\end{align*}
We skip the dual proof for the upper bound $r_i[t] < d_{\max}$.
\end{proof}
\begin{remark}\label{remark:quantizer}
Theorem~\ref{thm:y-safe} suggests that
we should carefully design the closest and farthest distances, $\min_j r_j^*$ and $\max_j r_j^*$,
in the target state because they constrain the set of safe initial states.
For example, when the distance in the target state $r_j^*$ is close to the minimum safe distances $d_{\min}$,
then the bound becomes tighter,
and hence fewer initial states can ensure safety.
In addition, $Q$ cannot be too coarse, that is, the value of the step radius $a$ should not be too large.
It should ensure $\frac{Q(d_{\min})}{\min_j \{r_j^*\}} < \frac{Q(d_{\max})}{\max\limits_j \{r_j^*\}}$.
Otherwise, $\mathrm{S}_0$ becomes an empty set, and no initial states can ensure safety.
\end{remark}

\subsection{Bound on Expected Convergence Time}\label{subsec:timebound}

Now, we analyze the upper bound on the convergence time.
We first provide the existing result on the convergence time of quantized consensus algorithms.
Then, we derive the upper bound for the formation control system.

Following~\cite{kashyap2007quantized}, the probability distribution of the convergence time is defined as
$T_{con}(\statez) = \inf \{ t \mid \statez[t] \in \xi, \statez[0]=\statez \}$.
Since our z-system evolves by a quantized gossip algorithm over \emph{linear networks},
the upper bound on the expected convergence time is provided in~\cite[Lemma 7]{kashyap2007quantized} as:
\begin{equation}\label{eq:convergence-bound}
    \max_{\statez \in \mathrm{Z}_0} \mathbb{E}[T_\mathit{con}(\statez)] \leq \frac{(z_{\max}-z_{\min})^2}{8} \cdot \frac{N(N^2-1)(N-1)}{4}
\end{equation}
where $z_{\min}$ and $z_{\max}$ are parameters specifying the minimum and maximum integer values among all possible states,
and $\mathrm{Z}_0=\{(z_1, \dotsc, z_N) \mid z_{\min} \leq z_i \leq z_{\max}\}$ includes all possible initial states.
With the upper bound in Equation~\eqref{eq:convergence-bound},
we derive the bound on the expected convergence time of our formation control system.
\begin{theorem}
\label{thm:conv-time}
The formation control system in Equation~\eqref{eq:controller_yi} converges with the following upper bound on the expected convergence time:
\[
\max_{\statey \in (\mathrm{Y_0} \cap \Safe_0)} \mathbb{E}[T_\mathit{con}(\statey)] \leq 
\frac{\Delta^2}{8} \cdot \frac{N(N^2-1)(N-1)}{4}
\]
where $\Delta = \max(I(d_{\max}) - I(d_{\min}), M - 1)$, $\mathrm{Y_0}$ is the same in Theorem~\ref{thm:y-eq}, and $\Safe_0$ is the same in Theorem~\ref{thm:y-safe}.
\end{theorem}
\begin{proof}
The proof is to apply Equation~\ref{eq:convergence-bound} and select $\Delta$ to be the larger upper bound among the bounds on $z_{\max} - z_{\min}$ for the radial and angular coordinates in the integer domain.
For the radial coordinate,
we derive a lower bound for $z_{\min} = \min_i \{I(r_i) - I(r_i^*)\}$ from $\statey \in \Safe_0$ as follows:
\begin{align*}
            & \bigwedge_{i=1}^N \frac{Q(d_{\min})}{\min\limits_j \{r_j^*\}}  < \frac{Q(r_i)}{r_i^*} \ 
\Rightarrow\  \frac{Q(d_{\min})}{\min\limits_j \{r_j^*\}} < \min_i \{\frac{Q(r_i)}{r_i^*}\} \\
\Rightarrow\ & I(d_{\min}) - \min\limits_j \{I(r_j^*)\} < \min_i \{I(r_i) - I(r_i^*)\} = z_{\min}
\end{align*}
Dually, we find the bound $z_{\max} < I(d_{\max}) - \max_j \{I(r_j^*)\}$.
Using the fact that $\max_j \{I(r_j^*)\} \geq \min_j \{I(r_j^*)\}$,
this leads to $z_{\max} -  z_{\min} < I(d_{\max}) - I(d_{\min})$.

For the angular coordinate,
the value of $I(\theta_i) - I(\theta_i^*)$ is bounded by 0 and $M-1$ because of the modulo operator,
hence the upper bound on $z_{\max} - z_{\min}$ is $M - 1$.
\end{proof}

\section{Quantization as Perception Contracts}\label{sec:perc-error}
For vision-based perception, uniform worst case bounds on the perception error between the ground truth and the perceived value can be overly conservative for system-level analysis. 
Recent research has shown that state-dependent error models can strike a balance between the conservatism of the safety analysis and the precision of characterizing deep learning-based perception systems~\cite{hsieh2022aap}. 

Following the same insight, we investigate the relationship between the ground truth and the perceived relative positions,
and we study how to search for the parameter values for the quantizer according to the empirically observed perception errors.
We randomly sampled pairs of camera images from two drones under different relative positions in \AirSim.
We fixed drone $i-1$ as the origin and uniformly sampled 10,000 positions of drone $i$ within a radius between 2 m to 20 m in the \AirSimNH environment from \AirSim.
For each sample, we obtained a pair of true relative position $y_i$ from \AirSim and perceived relative position $\hat{y}_{i}$ via vision-based pose estimation pipeline (of Section~\ref{subsec:prelim-vision}).

\begin{figure}[t]
    \centering
    \begin{subfigure}{0.40\columnwidth}
        \includegraphics[width=\textwidth,trim=0 5mm 0 3mm,clip]{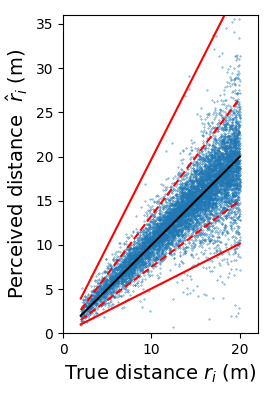}
        \caption{Radii on a linear scale.\label{subfig:radius-lin}}
    \end{subfigure}
    \begin{subfigure}{0.40\columnwidth}
        \includegraphics[width=\textwidth,trim=0 5mm 0 3mm,clip]{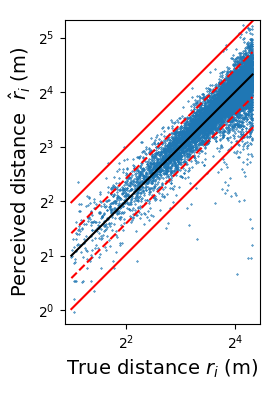}
        \caption{Radii on a log scale.\label{subfig:radius-log}}
    \end{subfigure}

    \begin{subfigure}[t]{0.40\columnwidth}
        \includegraphics[width=\textwidth,trim=0 15mm 0 17mm,clip]{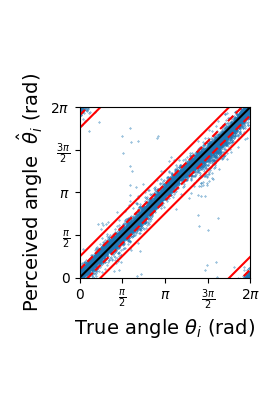}
        \caption{Angles in $[0,2\pi)$.\label{subfig:angle-lin}}
    \end{subfigure}
    \begin{subfigure}[t]{0.40\columnwidth}
        \includegraphics[width=\textwidth,trim=0 15mm 0 17mm,clip]{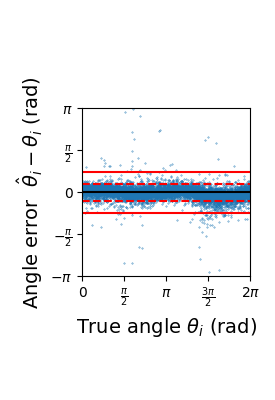}
        \caption{Angles in $[0,2\pi)$ and angle errors in $[-\pi,\pi)$.\label{subfig:angle-wrap}}
    \end{subfigure}

    \caption{Perceived relative positions $\hat{y}_{i}=[\hat{r}_i\,\hat{\theta}_i]^T$ with respect to true relative position $y_{i}=[r_i\,\theta_i]^T$.
    Sampled data are~\textcolor{pyplotblue}{blue dots}.
    The perceived value with no error is the \textbf{black line}.
    The empirical bounds are plotted as \textcolor{pyplotred}{red lines} such that
    99\% (90\%) of dots fall within the solid lines (the dashed lines).
    \label{fig:perc-err-bound-fun}}
\end{figure}

Figure~\ref{fig:perc-err-bound-fun} plots the perceived position $\hat{y}_{i}=[\hat{r}_i\,\hat{\theta}_i]^T$ with respect to
the true position $y_i=[r_i\,\theta_i]^T$.
In Figure~\ref{subfig:radius-lin}, we observe that the perceived distances $\hat{r}_i$ scatter wider when the true distance $r_i$ increases.
Secondly, much more perceived distances deviate greatly from the true distance when the true distance crosses a certain threshold,
e.g., about 15 meters in Figure~\ref{subfig:radius-lin}.
This is not too surprising:
As the two drones become farther apart, the overlap of the two camera views is smaller
and causes fewer matched features than eight pairs, which leads to inaccuracy in relative pose estimation.

Our ultimate goal is to design a quantizer whose quantization error overapproximates the perception error of the vision component,
and we empirically approximate the perception error from collected samples.
Recall in Section~\ref{ssec:quant}, the PC enforced by the quantizer $Q$ is $Q(\hat{y}_{i}) = Q(y_i)$.
We further simplify the PC as $Q(\hat{y}_{i}) = y_i$ because $y_i$ is a quantized value.
We start by expanding the definitions in Equation~\ref{def:quantized-radius} for the radial coordinates.
The PC can be rewritten as the following:
\begin{align*}
 & Q(\hat{r}_i) = r_i \Leftrightarrow a^{\round{\log_a \hat{r}_i}} = r_i \Leftrightarrow \round{\log_a \hat{r}_i} = \log_a r_i \\
   \Leftrightarrow\ & \log_a r_i-0.5 \leq \log_a \hat{r}_i < \log_a r_i+0.5 \\
   \Leftrightarrow\ & \log r_i -0.5\cdot\log a \leq \log \hat{r}_i < \log r_i + 0.5\cdot\log a
\end{align*}
This says that the PC defines a pair of \emph{linear bounds} around the ground truth on a log scale.
We can decrease or increase the value of $a$ to make the PC more strict or relaxed.
Ideally, the PC should hold for all observed data,
but this can lead to an overly relaxed PC that the quantizer $Q$ cannot ensure safety (See Remark~\ref{remark:quantizer}).
In practice, we may preprocess the data to remove outliers.
For example, the pair of red solid lines in Figure~\ref{subfig:radius-log} depicts
the bounds inferred from ignoring the worst 1\% of perceived values and therefore covering 99\% of data,
and the pair of red dashed lines depicts the bounds covering 90\% of data.
We also plot the bounds transformed back to the linear scale in Figure~\ref{subfig:radius-lin}.
Similarly, selecting larger or smaller $\theta_b$ defines more strict or relaxed \emph{constant bounds} for the angular coordinates.
It is worth noting that, due to the normalization,
angles are wrapped around 0 and $2\pi$ as shown in Figure~\ref{subfig:angle-lin}.
Therefore, we follow the standard approach to normalize the angle error $\hat{\theta_i} - \theta_i$ to $[-\pi, \pi)$ and infer the bounds as shown in Figure~\ref{subfig:angle-wrap}.

\begin{figure}[t!]
    \includegraphics[width=\columnwidth]{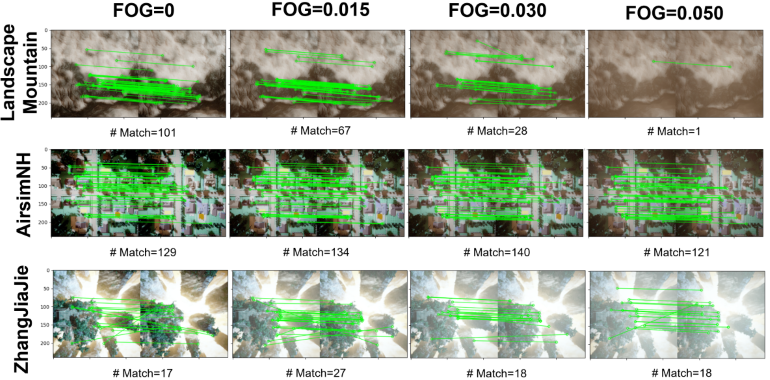}
    \centering
    \caption{Matched feature points found by the feature detection and matching algorithm under three AirSim environments and four fog levels.
}
    \label{fig:matched-feature-pairs}
\end{figure}

\begin{figure}[t!]
    \centering
    \includegraphics[width=\columnwidth, trim=0 0 0 12mm, clip]{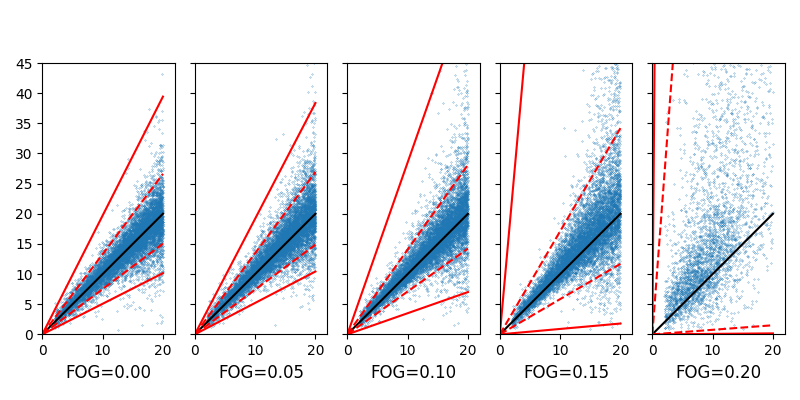}
    \caption{Perceived distances $\hat{r}_{i}$ and empirical linear bounds with respect to true distances $r_i$ under five fog levels.}
    \label{fig:perc-err-bound-varying-envs}
\end{figure}

The \PercCont also depends on environmental factors.
To systematically study the impact of environmental variations on the \PercCont, we experimented with different environments and weather conditions in the photorealistic \AirSim simulator.
Figure~\ref{fig:matched-feature-pairs} shows how the feature matching step degrades across three environments (namely \LandscapeMountains, \AirSimNH, and \ZhangJiajie) and four fog levels. 
Note that at the fog level 0.050, only one pair of matching features is detected for the same relative position under \LandscapeMountains.

\begin{figure}[t!]
    \includegraphics[height=3cm, trim=0 2mm 0 2mm, clip]{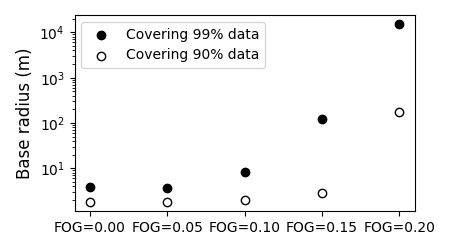}
    \centering
    \caption{Inferred step radius $a$ on a log scale covering 99\% and 90\% of data with respect to five fog levels. 
    }\label{fig:fog-percentile}
\end{figure}

Figure~\ref{fig:perc-err-bound-varying-envs} shows the \PercCont{}s for five fog levels under \AirSimNH.
The perception bound increases much faster (against the relative distance) in a foggier weather.
To better visualize this trend, we further plot the inferred step radius $a$ for covering 99\% and 90\% of data points under varying fog levels in Figure~\ref{fig:fog-percentile}.
Note that the y-axis is on a log scale in Figure~\ref{fig:fog-percentile},
so both values in fact increase faster than exponential growth with respect to the fog levels,
and the step radius for covering 99\% data increases more significantly.
Unsurprisingly, when it is too foggy, the value of $a$ is too large to satisfy the condition for Theorem~\ref{thm:y-safe} to ensure safety.

\section{Limitations and Discussions}
\label{sec:conclusion}
\begin{chiaoenv}
We presented an analysis for the convergence and safety of a vision-based formation control system.
To tackle the vagaries of the perception component,
our approach uses a \PercCont represented as a quantizer.
This quantizer captures the worst perception error in relative position estimates from the vision component,
which is then used to prove that the drones are safely separated and converges to the desired formation.
Especially, we designed non-uniform quantizers to model the \emph{state-dependent perception error}.
We empirically showed that a quantizer in log-polar coordinates models the observed perception error more accurately using the high-fidelity simulator, AirSim.
We also systematically studied the impact of environmental variations on the \PercCont.
We inferred quantization step sizes according to data sampled under each environment
so that the instantiated quantizer better models the observed error under the environment.

Our study assumed that all drones run synchronously and exchange image feature descriptors instantly.
This is obviously an idealization.
Our analysis will work without this assumption by bounding the change in relative positions under a fixed communication delay.
We can model the change in relative positions as part of the perception error.

Finally, this paper suggests a broad research direction on connecting quantized control and discrete abstractions over the continuous state space~\cite{alur2000discrete}.
Both quantization and discrete abstractions are partitioning the state space,
but they are different in that operators for quantized values such as difference, averaging, maximum, and minimum are not necessarily available for discrete abstractions.
Relating discrete abstractions with quantization will allow us to reuse the theories in quantized control for formal safety analyses.
\end{chiaoenv}







\bibliographystyle{IEEEtran}
\bibliography{references,sayan1}

\end{document}